\documentclass[article,12pt]{elsarticle}

\usepackage{graphicx}
\usepackage{lineno}
\usepackage{algorithmic}
\usepackage{amsmath}
\usepackage{amssymb}
\usepackage{amsfonts}
\usepackage{color}
\usepackage[colorlinks=true,
            linkcolor=red,
            urlcolor=blue,
            citecolor=blue]{hyperref}
\usepackage{natbib}
\usepackage{algorithm} 
\usepackage[all]{xy}
\usepackage{color}
\usepackage{DEF-bold}
\usepackage{DEF-script}

\def\tmin{t_{\rm min}}
\def\tmax{t_{\rm max}}

\def\S{\mathcal{S}}

\def\Pr{\mathbb{P}}
\def\Ex{\mathbb{E}}
\def\Ind{\mathbb{I}}

\def\p{^\prime}

\def\d{{\rm d}}

\def\tmin{t^{\rm min}} 
\def\tmax{t^{\rm max}} 
\def\tobs{t^{\rm obs}}

\def\qmax{Q^{\rm max}} 
\def\lmax{L^{\rm max}} 
\def\jump{J}
\def\ss{s^*}

\def\sss{s^{\dagger}}
\biboptions{authoryear}

\newtheorem{theorem}{Theorem}
\newtheorem{lemma}{Lemma}[theorem]
\newtheorem{proposition}{Proposition}[theorem]
\newproof{proof}{Proof}
\newdefinition{remark}{Remark}[theorem]
\begin{document}

\journal{}

\begin{frontmatter}

\title{Geometric ergodicity of Rao and Teh's algorithm  for Markov jump processes}

\author[mim]{B{\l}a{\.z}ej Miasojedow}
\ead{bmia@mimuw.edu.pl}
\author[mim]{Wojciech Niemiro}
\ead{wniem@mimuw.edu.pl}
\address[mim]{Institute of Applied Mathematics, University of Warsaw\\ 
Banacha 2, 02-097 Warsaw,  Poland}
\begin{abstract}
\citet{RaoTeh2013a} 
introduced an efficient MCMC algorithm for sampling from the posterior distribution of a hidden 
Markov jump process. The algorithm is based on the idea of sampling virtual jumps. In the present paper we show that the
Markov chain generated by Rao and Teh's algorithm is geometrically ergodic. To this end we establish a geometric drift condition towards a small set.
\end{abstract}


\end{frontmatter}

\section{Introduction}
Markov jump processes (MJP) are natural extension of Markov chains to continuous time. 
They are widely applied in modelling of the phenomena of chemical, biological,
economic and other sciences.  

In many applications it is necessary to consider a situation where the trajectory of a MJP 
is not observed directly, only partial and noisy observations are available. Typically, 
the posterior distribution over trajectories is then analytically intractable. In the literature there exist several approaches 
to the above mentioned problem: based on sampling 
\citep{BoysWilkKirk2008,EFK,FaSh,FarSher2006,GolWilk2011,GolWilk2014,GolHendSher2015,Nod2,RaoTeh2013a,rao2012mcmc}, and also based on 
numerical approximations. 
 To the best of our knowledge the most general, efficient and exact method for a finite state space is that proposed by 
\citet{RaoTeh2013a}, and extended to a more general class
of continuous time discrete systems in \citet{rao2012mcmc}. 

In the present paper we establish geometric ergodicity of Rao and Teh's algorithm for homogeneous MJPs.
Geometric ergodicity is a key property of Markov chains which allows to use Central Limit Theorem for sample averages.

The rest of the paper is organised as follows. In Section~\ref{sec:intro} we briefly introduce hidden Markov jump processes, 
next in Section~\ref{sec:rao} we recall the Rao and Teh's algorithm.
The main result is proved in Section~\ref{sec:main}.

\section{Hidden Markov jump processes}\label{sec:intro}

Consider a continuous time homogeneous Markov process $\{X(t),{\tmin\leq t\leq \tmax}\}$ on a finite state space $\S$.
Its probability law is defined via the initial distribution $\nu(s)=\mathbb{P}(X(\tmin)=s)$ and the transition intensities 
\begin{equation*}
  Q(s,s^\prime)=\lim_{h\to 0}\frac{1}{h}\Pr(X(t+h)=s\p|X(t)=s)
\end{equation*}
for $s,s^\prime\in\S$, $s\not=s\p$. Let $Q(s)=\sum_{s^\prime\neq s} Q(s,s^\prime)$ denote the intensity of leaving state $s$. 
For definiteness, assume that $X$ has right continuous trajectories. We say $X$ is a Markov jump process (MJP).

Suppose that process $X$ cannot be directly observed but we can observe some random
quantity $Y$ with probability distribution $L(Y|X)$. Let us say $Y$ is the evidence
and $L$ is the likelihood. The problem is to restore the hidden trajectory of $X$ given $Y$.  From the
Bayesian perspective, the goal is to compute/approximate/sample from the posterior 
\begin{equation*}
 p(X|Y)\propto p(X)L(Y|X).
\end{equation*}
Function $L$, transition probabilities $Q$ and initial distribution $\nu$ are assumed 
to be known. To get explicit form of posterior distribution we consider a typical form of noisy observation. 
Assume that the trajectory $X([\tmin,\tmax])$ is observed independently at $k$  deterministic time points with
some random errors. Formally, we observe $Y=(Y_1,\ldots,Y_k)$ where
\begin{equation}\label{eq: obs}
 L(Y|X)=\prod_{j=1}^k L_j(Y_j|X(\tobs_j)),
\end{equation}
for some fixed known points $\tmin\leq \tobs_1<\cdots<\tobs_r\leq \tmax$.  
 
\section{Uniformization and Rao and Teh's algorithm}\label{sec:rao}

The so-called ``uniformization'' is a classical and well-known representation of a Markov jump process in terms
of potential times of jumps and the corresponding states \citep{jensen1953markoff}. Every  trajectory 
$X([\tmin,\tmax])$ is right continuous and piece-wise constant:  $X(t)=S_{i-1}$ for $T_{i-1}\leq t< T_i$, 
where random variables $T_i$ are such that $\tmin< T_1<\cdots<T_N<\tmax$ (by convention, $T_0=\tmin$ and $\tmax<T_{N+1}$).  
Random sequence of states $S=(S_1,\ldots,S_N)$ such that $S_i=X(T_i)$ is called a skeleton. 
We do not assume that $S_{i-1}\not=S_i$, and therefore the two sequences 
\begin{equation}\nonumber
 \begin{pmatrix} 
  T\\S
 \end{pmatrix}
 =
 \begin{pmatrix} \tmin & T_1 & \cdots &  T_i & \cdots & T_N & \tmax\\ 
S_0 & S_1 & \cdots & S_i & \cdots & S_N & \end{pmatrix}.
\end{equation}
represent the process $X$ in a redundant way: many pairs $(T,S)$ correspond to the same trajectory $X([\tmin,\tmax])$. 
Let $J=\{i \in [1\colon N]: S_{i-1}\neq S_{i}\}$, so that $T_J=(T_i:i \in J)$ are moments of \textit{true} jumps and
$T_{-J}=T\setminus T_J=(T_i:i \not\in J)$ are \textit{virtual} jumps.  {We write $[l:r]=\{l,l+1,\ldots,r\}$.}
By a harmless abuse of notation, 
we identify increasing sequences of points in $[\tmin,\tmax]$ with finite sets. Note that the trajectory of $X$ is
uniquely defined by $(T_J,S_J)$. Let us write $X\equiv (T_J,S_J)$ and also use the notation $\jump(X)=T_J$ for the set of
true jumps.

Uniformization obtains if $T$ is a (sequence of consecutive points of) a homogeneous Poisson process  with intensity $\lambda$, 
where $\lambda\geq \qmax=\max_sQ(s)$. The skeleton $S$ is then (independently of $T$) a discrete time homogeneous Markov chain with
the initial distribution $\nu$ and the transition matrix
\begin{equation}
 \label{eq: P_uniformization}
 P(s,s^\prime)=\begin{cases}
                \displaystyle\frac{Q(s,s^\prime)}{\lambda}&\text{  if } s\neq s^\prime\;;\\
                \\
                1- \displaystyle\frac{Q(s)}{\lambda}&\text{ if } s= s^\prime\;.
               \end{cases}
\end{equation}

\citet{RaoTeh2013a} exploit uniformization to construct a special version of Gibbs sampler which converges
to the posterior $p(X|Y)$. The key facts behind their algorithm are the following. First, given the trajectory $X\equiv
(T_J,S_J)$ the conditional distribution of virtual jumps $T_{-J}$ is that of the non-homogeneous (actually piece-wise homogeneous)
Poisson process with intensity $\lambda-Q(X(t))\geq 0$. Second, this distribution does not change if we introduce the likelihood.
Indeed, $L(Y|X)=L(Y|T_J,S_J)$, so $Y$ and $T_{-J}$ are conditionally independent given $(T_J,S_J)$ and thus $p(T_{-J}|T_J,S_J,Y)=p(T_{-J}|T_J,S_J)$. Third,
the conditional distribution $p(S|T,Y)$ is that of a hidden discrete time Markov chain and can be efficiently sampled from using
the algorithm FFBS (Forward Filtering-Backward Sampling). 

The Rao and Teh's algorithm generates a Markov chain $X_0,X_1,\ldots,X_m,\ldots$ (where $X_m=X_m([\tmin,\tmax])$ is a trajectory of MJP), convergent
to $p(X|Y)$. A single step, that is the rule of transition from $X_{m-1}=X$ to $X_{m-1}=X\p$ is described in Algorithm 1. 
\begin{algorithm}\label{alg: RT}
 \caption{Single step of Rao and Teh's algorithm.}
 \begin{algorithmic}
  \STATE {\bf input:} previous state $(T_J,S_J)\equiv X$ and observation $Y$.
  \STATE {(V)} Sample a Poisson process $V$ with intensity $\lambda-Q(X(t))$ on $[\tmin,\tmax]$. Let $T\p=T_J\cup V$ 
  \COMMENT{\textit{new set of potential times of jumps}}.
  \STATE {(S)}  Draw new skeleton $S\p$ from the conditional distribution $p(S\p|T\p,Y)$ by FFBS.  
  The new allocation of virtual and true jumps is via $J\p=\{i: S_{i-1}\p\neq S_{i}\p\}$ 
  \COMMENT{\textit{we discard new virtual jumps $T\p_{-J\p}$}}.
  \RETURN new state  $(T\p_{J\p},S\p_{J\p})\equiv X\p$.
  \end{algorithmic}
\end{algorithm}

Convergence of the algorithm has been shown by its authors in \citet{RaoTeh2013a}. It follows from the fact that the chain has the stationary distribution
$p(X|Y)$ and is irreducible and aperiodic, provided that $\lambda>\qmax$. 

\section{Main result}\label{sec:main}

Consider the Markov chain  $X_0,X_1,\ldots,X_m,\ldots$ generated by the Rao and Teh's algorithm. 
Its transition kernel is denoted by $A$. Put differently, $A(X,\d X\p)$ is the probability distribution corresponding to Algorithm 1.
Let $\Pi(\d X)$ be the posterior distribution of $X$ given $Y$.  In this paper we consider only Monte Carlo randomness, so $Y$ is fixed 
and can be omitted in notation.  The obvious standing assumption is that $L(Y|X)>0$ happens with nonzero probability if $X$ is given by $\nu$ and $Q$. 
It just means that the hidden MJP under consideration is ``possible''. 

\begin{theorem}\label{th: main}
Assume that the matrix of intensities $Q$ is irreducible and $\lambda>\qmax$. Then the chain is geometrically ergodic. There exist
constant $\gamma<1$ and function $M$ such that for every $X$,
\begin{equation}\nonumber
 \Vert A^m(X,\cdot)-\Pi(\cdot)\Vert_{\rm tv} \leq \gamma q^m M(X).
\end{equation}
\end{theorem}

We begin with some auxiliary results. 

\begin{lemma}\label{lem: Jfixends}
Let $S_0,S_1,\ldots,S_n$ be a Markov chain on a finite state space $\S$, with transition matrix $P$. Assume that $P$ is irreducible and
$P(s,s)\geq \eta>0$ for all $s\in\S$. Let $J=\{i\in [1\colon n] : S_i\not=S_{i-1}\}$. There exist $n_0$ and $\delta>0$ such that for every 
$n\geq n_0$ and for all $s_0,s_n\in\S$,
\begin{equation}\nonumber
 \Ex(|J||S_0=s_0,S_n=s_n)\leq (1-\delta)n. 
\end{equation}   
\end{lemma}
\begin{proof} Under the assumptions of the lemma, the Markov chain is ergodic. Let $\pi$ be its (strictly positive) stationary distribution. 
Fix $\eps>0$.  There exists $n_0$ such that for $n\geq n_0/2-1$ we have for all $s_0,s_n\in\S$,
\begin{equation}\nonumber
(1-\eps)\pi(s_n)\leq \Pr(S_n=s_n|S_0=s_0)\leq (1+\eps)\pi(s_n). 
\end{equation}    
Now if $n\geq n_0$ then we can  choose $n_1$ such that $n_1\geq n_0/2$ and $n-n_1\geq n_0/2-1$. Let $J_{1}=\{i\in [1\colon n_1] : S_i\not=S_{i-1}\}$.
With the notation $\rho(j,s)=\Pr(|J_{n_1}|=j,S_{n_1}=s|S_0=s_0)$ we have
\begin{equation*}
\begin{split}
 \Ex(|J_{1}|\Ind(S_n=s_n)|S_0=s_0)&=\sum_j\sum_s j \rho(j,s )P^{n-n_1}(s,s_n)\\
                                    &\leq \sum_j\sum_s j \rho(j,s) (1+\eps)\pi(s_n)\\
                                    &=\Ex(|J_{n_1}||S_0=s_0) (1+\eps)\pi(s_n)\\
                                    &\leq (1-\eta)n_1 (1+\eps)\pi(s_n).
\end{split}
\end{equation*}
The last inequality follows from the fact that $\Pr(S_i=s|S_{i-1}=s)\geq \eta$. Consequently,
\begin{equation}\nonumber
\begin{split}
\Ex(|J_{1}||S_n=s_n,S_0=s_0)&=\frac{\Ex(|J_{1}|\Ind(S_n=s_n)|S_0=s_0)}{\Pr(S_n=s_n|S_0=s_0)}\\
                              &\leq   (1-\eta)\frac{1+\eps}{1-\eps}n_1=(1-2\delta)n_1,  
\end{split}
\end{equation}
where $\delta>0$ if $\eps$ is chosen sufficiently small. Finally, since $|J|\leq |J_{1}|+n-n_1$ and $n_1\geq n/2$, we obtain
\begin{equation}\nonumber
\Ex(|J||S_n=s_n,S_0=s_0)\leq (1-2\delta)n_1 +n-n_1=n-2\delta n_1\leq (1-\delta)n. 
\end{equation}
We get the conclusion.
\end{proof}

\begin{lemma}\label{lem: Jfixobs}
Let the assumptions of Lemma \ref{lem: Jfixends} hold. Numbers $n_0$ and $\delta$ are the same as in this lemma. Let $k$ be fixed. 
If $n\geq (k+1)n_0$, then for arbitrarily chosen indices
$0\leq i_1\leq i_2\leq\cdots\leq i_k\leq n$ and states $s_1,\ldots,s_k$ we have that 
\begin{equation}\nonumber
 \Ex(|J||S_{i_1}=s_1,\ldots,S_{i_k}=s_k)\leq \left(1-\frac{\delta}{k+1}\right)n,
\end{equation}   
provided that $\Pr(S_{i_1}=s_1,\ldots,S_{i_k}=s_k)>0$.
\end{lemma}
\begin{proof} Write $n_j=i_j-i_{j-1}$, where by convention $i_0=0$ and $i_{k+1}=n$. Let $J^{(j)}=\{i\in [i_{j-1}+1\colon i_j] : S_i\not=S_{i-1}\}$
for $j=1,\ldots,k+1$. (Let us mention that we have not excluded the case when $i_{j-1}=i_j$   for some $j$ and thus $J^{(j)}=\emptyset$.
Of course, then we must have $s_{j-1}=s_j$ in the conditional expectation.) If $n\geq (k+1)n_0$ then there is at least one $j$ such that 
$n_j\geq n/(k+1)\geq n_0$, because $\sum_j n_j=n$. From Lemma  \ref{lem: Jfixends} we infer that for this $j$ we have
\begin{equation}\nonumber
 \Ex(|J^{(j)}||S_{i_1}=s_1,\ldots,S_{i_k}=s_k)= \Ex(|J^{(j)}||S_{i_{j-1}}=s_{j-1},S_{i_j}=s_j)\leq (1-\delta)n_j. 
\end{equation} 
The equality above is true because of the (two-sided) Markov property.  Analogously as in the previous proof, we observe that
$|J|\leq |J^{(j)}|+n-n_j$, so
\begin{equation}\nonumber
\Ex(|J||S_{i_1}=s_1,\ldots,S_{i_k}=s_k)\leq (1-\delta)n_j +n-n_j=n-\delta n_j\leq \left(1-\frac{\delta}{k+1}\right)n. 
\end{equation}
\end{proof}


In the next proposition we establish a geometric drift condition for the Markov chain  $X_0,X_1,\ldots,X_m,\ldots$.
Consider a single step, that is transition from $X_{m-1}=X$ to $X_{m}=X\p$. Thus 
only Monte Carlo randomness in Algorithm 1 is taken into account. The dependence on the input trajectory $X$ (and on $Y$) is implicitly 
assumed but indicated only when necessary.   
Recall that $|\jump(X)|$ is the number of true jumps of the trajectory $X([\tmin,\tmax])$. 

\begin{proposition}\label{pr: drift}
There exist $q<1$ and $c<\infty$ such that in a single step of the Rao and Teh's algorithm we have 
$\Ex(|\jump(X\p)||X)\leq q|\jump(X)|+c$.   
\end{proposition}

\begin{proof} Let us analyse what happens in both two stages of Algorithm 1. We fix the initial $X$. In the first
stage we add a new set $V$ of potential jumps. Since
$|V|$ has the Poisson distribution with intensity $\int_{\tmin}^{\tmax} (\lambda-Q(X(t))\d t$, we have 
$\Ex|V|\leq \mu=\lambda(\tmax-\tmin)$. Thus we obtain $T\p$ with 
$\Ex|T\p|\leq |\jump(X)|+\mu$. In the second stage $T\p$ is ``thinned'' to $T\p_{J\p}$. We will prove that
\begin{equation}\label{eq: thin}
 \Ex(|J\p||T\p)\leq \left(1-\frac{\delta}{k+1}\right)|T\p|, \text{ if } |T\p|\geq (k+1) n_0. 
\end{equation}
Inequality \eqref{eq: thin} implies $\Ex(|\jump(X\p)|)=\Ex(|J\p|)\leq (1-\delta/(k+1))(|\jump(X)|+\mu)$. The conclusion
of the theorem follows with $q=(1-\delta/(k+1))$. To ensure that the conclusion holds also if $|T\p|<(k+1) n_0$, we can choose $c=q\mu+ (k+1) n_0$
so that $|J\p|<c$.
 
It remains to show \eqref{eq: thin}. We consider the second stage of Algorithm 1 
(sampling a new skeleton $S\p$ from $p(S\p|T\p,Y)$). From now on, the result of the first stage (updating $T$ to $T\p$)  
is fixed. 
Note that our assumption about the structure of observations \eqref{eq: obs} implies that  
$p(S\p|T\p,Y)\propto  p(S\p|T\p)\prod_{j=1}^r L_j(Y_j|S\p_{i_j})$, where $i_j=\max\{i:T\p_i\leq \tobs_j\}$. 
Since $p(S\p|T\p)$ is the distribution of a Markov chain, it follows that
\begin{equation}\nonumber
 p(S\p|T\p,Y)\propto  \nu(S\p_0)\prod_{i=1}^{|T\p|} P(S\p_{i-1},S\p_i) \prod_{j=1}^k L_j(S\p_{i_j}),
\end{equation} 
where $P$ is given by \eqref{eq: P_uniformization} and $L_j(s)=L_j(Y_j|s)$. 

Although the actual sampling from $p(S\p|T\p,Y)$ is by FFBS, exactly the same result can be obtained via rejection sampling as follows. 
\begin{itemize}
 \item[(S1)] Simulate Markov chain $S\p$ (of length $|T\p|$) with transition matrix $P$ given by \eqref{eq: P_uniformization} and initial distribution $\nu$.
 \item[(S2)] The skeleton $S\p$ is accepted with probability $\prod_{j=1}^k (L_j(S\p_{i_j})/\lmax_j)$, with some $\lmax_j\geq \max_s L_j(s)$.
 Otherwise go to (S1).
\end{itemize}
(Of course the rejection method is highly inefficient and is considered only to clarify presentation.)
Now we are in a position to use Lemma \ref{lem: Jfixobs}, {with $\eta=\lambda-\qmax$}. If $|T\p|\geq (k+1) n_0$ then
\begin{equation}\nonumber
\begin{split}
 & \Ex(|J\p||T\p)=\Ex(|J\p||S\p\text{ accepted})\\
 & = \sum_{s_1,\ldots,s_k}\Ex(|J\p||S\p_{i_1}=s_1,\ldots,S\p_{i_k}=s_k)\Pr(S\p_{i_1}=s_1,\ldots,S\p_{i_k}=s_k|S\p\text{ accepted})\\
 & \leq \left(1-\frac{\delta}{k+1}\right)|T\p|  \sum_{s_1,\ldots,s_k}\Pr(S\p_{i_1}=s_1,\ldots,S\p_{i_k}=s_k|S\p\text{ accepted})\\
 & \leq \left(1-\frac{\delta}{k+1}\right)|T\p|
\end{split}
\end{equation} 
(of course in the sum above we can omit ``impossible sequences'' $s_1,\ldots,s_k$).  We have proved \eqref{eq: thin} and we are done.
\end{proof}
{\begin{remark} The proof of Proposition \ref{pr: drift} could be significantly simpler if we assumed that the likelihood $L_j(s)$ is 
strictly positive. However, this assumption is not satisfied in many interesting applications (e.g. if the observations are without noise).
 \end{remark}}

\begin{proposition}\label{pr: small}
The set $\{X: |\jump(X)|\leq h\}$ is 1-small for every $h$. 
\end{proposition}
\begin{proof}
We are going to construct a skeleton $\ss$ and a set of the corresponding times of jumps $\mathcal{T}$ which are achievable with 
uniformly lower bounded probability, in one step of the algorithm,  from any trajectory $X$ such that $|\jump(X)|\leq h$.   
Let us first choose a sequence $\sss=(\sss_1,\sss_2,\dots,\sss_k)$ such that
\[\prod_{j=1}^k (L_j(\sss_j)/\lmax_j)=\beta_1>0\;.\]
By irreducibility of kernel $P$ we can extend $\sss$ to a possible skeleton $\ss$, i.e.\ we define a sequence $\ss=(\ss_0,...,\ss_l)$ for some  
$l\geq k$ such that $\sss$ is a subsequence of $\ss$, $\ss_{i-1}\not=\ss_i$ and 
\[\nu(\ss_0)\prod_{i=1}^{l} P(\ss_i,\ss_{i+1})= \beta_2>0\;.\]
We denote by $i_1,\dots,i_k$ indices of elements of $\sss$ in $\ss$. Now we define a set $\mathcal{T}$ as follows.
\[\mathcal{T}=\{(t_1,...,t_l)\;:\; t_{i_j}\leq\tobs_j,\ t_{i_j+1}>\tobs_j\text{ for } j=1,\dots,k\}\;.\]

Fix $X$ with  $|\jump(X)|\leq h$. Let us first describe a special way in which Algorithm 1 can be executed. 
In stage (V) we can independently sample two Poisson processes on the interval $[\tmin,\tmax]$, say $V^0$  and  $V^{\rm rest}$,
with intensities $\lambda-\qmax$ and $\qmax-Q(X(t))$, respectively. Next set $V=V^0\cup V^{\rm rest}$ and $T\p=\jump(X)\cup V$. Note that 
\[\Pr(V^0\in\mathcal{T})=\beta_3>0\;.\]
In stage (S) of the Algorithm 1 we construct skeleton $S\p$. 
Assume that we use rejection sampling, defined by (S1) and (S2) in the previous proof. We use the notation as in Algorithm 1. 
Recall that $J\p=\{i: S_{i-1}\p\neq S_{i}\p\}$.
We will bound from below the probability that 
at stage (S1) all points belonging to $V^{\rm rest}\cup\jump(X)$ are changed to virtual jumps, while jumps at $V^0$ form the skeleton $\ss$.
We have
\begin{equation}\nonumber
\begin{split}
 \Pr(T\p_{J\p}= V^0,S\p_{J\p}=\ss|T\p)&\geq\beta_2\sum_{k=0}^\infty \eta^{|\jump(X)|+k}\Pr(V^{\rm rest}=k|X)\\
                                      &= \beta_2\Ex (\eta^{h+|V^{\rm rest}|}|X)\geq \beta_2\eta^{h+\Ex(|V^{\rm rest}||X)}                        
\end{split} 
\end{equation}
by Jensen inequality. Since $V^{\rm rest}$ is a Poisson process with intensity bounded by $\qmax$ we have 
$\Ex(|V^{\rm rest}||X)\leq\qmax(\tmax-\tmin)$, therefore we conclude that $\Pr(T\p_{J\p}= V^0,S\p_{J\p}=\ss|T\p)\geq \beta_4>0$. 
The probability of accepting the obtained skeleton $S^\prime=\ss$ at stage (S2) is clearly equal to $\beta_1$. 
{We have shown that $\Pr(V^0\in\mathcal{T},T\p_{J\p}= V^0, S\p_{J\p}=\ss|X)\geq\beta_1\beta_4\beta_3=\beta>0$ 
whenever $|\jump(X)|\leq h$. 
Moreover, if $T\p_{J\p}= V^0$ and $S\p_{J\p}=\ss$ then  the 
probability distribution of $(T\p_{J\p},S\p_{J\p})\equiv X\p$ does not depend on $X$.}
This completes the proof.
\end{proof}

\newproof{proofmain}{Proof of Theorem \ref{th: main}}
\begin{proofmain}
Theorem \ref{th: main} immediately follows from Propositions \ref{pr: drift} and \ref{pr: small}, see for example \cite[Th. 9]{roberts2004general}. 
\end{proofmain}
\bibliography{refs}

\end{document}